\documentclass [11pt]{amsart}
\usepackage{amsmath}
\usepackage{amssymb}
\usepackage{amscd}
\usepackage{mathabx}
\usepackage{tikz}
\usetikzlibrary{matrix}

\usepackage{scalerel,stackengine}
\usepackage{graphicx}

\stackMath
\newcommand\reallywidehat[1]{%
\savestack{\tmpbox}{\stretchto{%
  \scaleto{%
    \scalerel*[\widthof{\ensuremath{#1}}]{\kern-.6pt\bigwedge\kern-.6pt}%
    {\rule[-\textheight/2]{1ex}{\textheight}}
  }{\textheight}%
}{0.5ex}}%
\stackon[1pt]{#1}{\tmpbox}%
}

\newcommand\reallywidecheck[1]{%
\savestack{\tmpbox}{\stretchto{%
  \scaleto{%
    \scalerel*[\widthof{\ensuremath{#1}}]{\kern-.6pt\bigwedge\kern-.6pt}%
    {\rule[-\textheight/2]{1ex}{\textheight}}
  }{\textheight}%
}{0.5ex}}%
\stackon[1pt]{#1}{\scalebox{-1}{\tmpbox}}%
}

\allowdisplaybreaks

\newcommand{\oplam}{\mbox{\LARGE $\curlywedge$}}

\numberwithin{equation}{section}

\newcommand{\supp}{\mbox{\rm supp}}
\newcommand{\dens}{\mbox{\rm dens}}

\newcommand{\RR}{{\mathbb R}}

\newcommand{\ZZ}{{\mathbb Z}}
\newcommand{\CC}{{\mathbb C}}

\newcommand{\KK}{\mathbb K}

\newcommand{\cL}{{\mathcal L}}

\newcommand{\cK}{{\mathcal K}}
\newcommand{\cM}{{\mathcal M}}
\newcommand{\cS}{{\mathcal S}}

\newcommand{\dd}{\mbox{\rm d}}

\newcommand{\Cu}{C_{\mathsf{u}}}
\newcommand{\Cc}{C_{\mathsf{c}}}
\newcommand{\Cz}{C^{}_{0}}

\newcommand{\SAP}{\mathcal{SAP}}

\newtheorem{lemma}{Lemma}[section]
\newtheorem{theorem}[lemma]{Theorem}
\newtheorem{proposition}[lemma]{Proposition}
\newtheorem{corollary}[lemma]{Corollary}
\newtheorem{conjecture}[lemma]{Conjecture}

\theoremstyle{definition}
\newtheorem{definition}[lemma]{Definition}

\newtheorem{remark}{Remark}

\begin{document}

\title[Fourier Analysis on Meyer Sets]{On the Fourier Analysis of Measures with Meyer Set Support}

\author{Nicolae Strungaru}
\address{Department of Mathematical Sciences, MacEwan University \\
10700 -- 104 Avenue, Edmonton, AB, T5J 4S2\\
Phone: 780-633-3440 \\
and \\
Institute of Mathematics ``Simon Stoilow''\\
Bucharest, Romania}
\email{strungarun@macewan.ca}
\urladdr{http://academic.macewan.ca/strungarun/}

\begin{abstract} In this paper we show the existence of the generalized Eberlein decomposition for Fourier transformable measures with Meyer set support. We prove that each of the three components is also Fourier transformable and has Meyer set support. We obtain that each of the pure point, absolutely continuous and singular continuous components of the Fourier transform is a strong almost periodic measure, and hence is either trivial or has relatively dense support. We next prove that the Fourier transform of a measure with Meyer set support is norm almost periodic, and hence so is each of the pure point, absolutely continuous and singular continuous components. We show that a measure with Meyer set support is Fourier transformable if and only if it is a linear combination of positive definite measures, which can be chosen with Meyer set support, solving a particular case of an open problem. We complete the paper by discussing some applications to the diffraction of weighted Dirac combs with Meyer set support.
\end{abstract}

\maketitle

\section{Introduction}

The discovery of quasicrystals in the 1980's \cite{She} emphasized the need for a better understanding of the mathematics behind the physical process of diffraction and of aperiodic order in general.

A large class of mathematical models for quasicrystals is represented by Meyer sets and weighted Dirac combs with Meyer set support. It was believed that the Meyer property was the reason behind the unusual diffraction spectrum of quasicrystals \cite{BCG}, and this was proven to be the case: indeed Meyer sets show a large Bragg diffraction spectrum which is highly ordered \cite{NS1,NS2,NS11}.

The pure point Bragg spectrum of Meyer sets is now pretty well understood \cite{JBA,BM,MS,NS1,NS2,NS2,NS5,NS11}. We also know a little about the Bragg continuous spectrum \cite{JBA,NS1,NS5,NS11} but nothing is known in general about the absolutely continuous and singular continuous Bragg spectrum of Meyer sets. It is our goal in this project to study these two components of the diffraction spectrum of Meyer sets via a systematic study of the Fourier analysis of measures with Meyer set support.

\smallskip

As it was introduced by Hof \cite{HOF3}, the physical diffraction of a solid can be viewed as the Fourier transform $\widehat{\gamma}$ of the autocorrelation measure $\gamma$ of the structure. The measure $\gamma$ is positive definite, and often also positive, and therefore it is Fourier transformable as a measure \cite{BF,ARMA,MoSt}, and its Fourier transform $\widehat{\gamma}$ is a positive measure. As any measure, $\widehat{\gamma}$ has a Lebesgue decomposition
\begin{displaymath}
\widehat{\gamma}=\left(\widehat{\gamma} \right)_{pp}+\left(\widehat{\gamma} \right)_{ac}+\left(\widehat{\gamma} \right)_{sc}
\end{displaymath}
into a pure point, absolutely continuous and singular continuous component. We will refer to these three components as the \textbf{pure point, absolutely continuous} and \textbf{singular continuous diffraction spectra}, respectively, of our structure.

Structures with pure point diffraction, that is structures for which the absolutely and singular continuous diffraction spectra are absent, are now relatively well understood (see for example \cite{TAO,BHS,BL,BL2,BM,DL03,LS,LR,LO,LO1,Meyer1,MOO,NS11}). Of particular interest among pure point diffractive point sets are model sets, which are constructed by projecting points in a strip from a higher dimensional lattice (see Definition~\ref{CPS} below for the exact definition). If the window which produces the strip is regular, then the model set is pure point diffractive \cite{Hof2,Martin2,BM,CRS}. It was recently shown that the regularity of the window can be replaced by a weaker natural condition on the density of the pointset \cite{BHS,KR,KR2}. Recent work \cite{TAO,CRS} has also shown that the diffraction formula for (regular) model sets is just the Poisson summation formula for the underlying lattice, an idea which seems to have been first suggested by Lagarias (see \cite[Page 9]{BG}). This emphasizes that the long range order exhibited by model sets is a consequence of the periodicity of the lattice in the .

Moving beyond the pure point case, not too much is known in general. For 1-dimensional substitution tilings, some recent progress towards understanding the nature of the diffraction spectrum has been made \cite{BaGa}, but much still needs to be done.

\smallskip

A large class of point sets with long range order are Meyer sets. They have been introduced by Y. Meyer in \cite{Meyer}, and studied in \cite{JBA,LAG1,MOO,NS11}. They can be characterized via completely different properties: as Delone subsets of model sets, as almost lattices, as Delone sets with a uniform discrete Minkowski difference $\Lambda-\Lambda-\Lambda$ (in $\RR^d$ it suffices for $\Lambda- \Lambda$ to be uniformly discrete), as harmonious sets and as Delone sets with relatively dense sets of $\epsilon$-dual characters (see \cite{NS11} for the full characterisation in second countable LCAG, or \cite{LAG1,Meyer,MOO} for the characterisation for $G=\RR^d$).

As arbitrary subsets of model sets, one should expect Meyer sets to exhibit some of the inherited order: we expect a large pure point (or Bragg) spectrum and potentially some continuous diffraction spectrum, a consequence of the generic randomness introduced by going to arbitrary subsets. This is indeed the case: the diffraction pattern of a Meyer set has a relatively dense set of Bragg peaks \cite{NS1}, which are highly ordered \cite{NS2}. The pure point diffraction measure is a strongly almost periodic measure \cite{NS1}, which is also sup almost periodic \cite{NS2}. The continuous diffraction measure $(\widehat{\gamma})_{\mathsf{c}}=(\widehat{\gamma})_{ac}+(\widehat{\gamma})_{sc}$ is a strong almost periodic measure \cite{NS2}, and hence the continuous diffraction spectrum is either non-existent or has a relatively dense support. All these results have been generalized to weighted Dirac combs with Meyer set support \cite{NS5,NS11}. One would expect that the strong almost periodicity of the continuous spectrum would imply that both the singular continuous and absolutely continuous spectrum are strongly almost periodic, but this is not true. Recently, in \cite{SS}, the authors constructed examples $\mu, \nu \in \SAP(\RR)$ of continuous measures such that $\mu_{ac}$ and $\nu_{sc}$, respectively, are non trivial measures with compact support.

\smallskip
It is our goal in this paper to show that given a measure $\gamma$ with Meyer set support, then $(\widehat{\gamma})_{ac} \in \SAP(G)$ and $(\widehat{\gamma})_{sc}SAP(G)$. We prove this via a systematical study the Fourier transform of measures with Meyer set support. The main issue we face is the enigmatic nature of strong almost periodicity, which is compatible with the Lebesgue decomposition \cite{SS}. To solve this problem we show that the class of Fourier transformable measures with Meyer set support admits a generalized notion of \textbf{Eberlein decomposition} (see \cite{ARMA,MoSt} for definitions and properties): In Theorem~\ref{T2} we show that every Fourier transformable measure $\gamma$ supported inside a Meyer set $\Lambda$ has an unique decomposition
\begin{equation}\label{EQ3}
\gamma= \gamma_{s}+\gamma_{0a}+\gamma_{0s}
\end{equation}
into three Fourier transformable measures, each supported inside a Meyer set, such that
\begin{displaymath}
\reallywidehat{\gamma_{s}}=\left(\widehat{\gamma}\right)_{pp} \,;\, \reallywidehat{\gamma_{0a}}=\left(\widehat{\gamma}\right)_{ac} \,;\, \reallywidehat{\gamma_{0s}}=\left(\widehat{\gamma}\right)_{sc} \,.
\end{displaymath}
We obtain this result by exploiting the Meyer property of the support via an embedding into a model set, and by using the Fourier analysis of certain weighted Dirac combs defined by the underlying cut and project scheme (CPS). It is worth pointing out that the Fourier theory we use for the larger weighed Dirac cobs is a consequence of the Poisson summation formula of the Lattice $\cL$ in the CPS. Therefore, similarly to regular model sets, the existence of the generalized Eberlein decomposition (\ref{EQ3}), as well as all the strong consequences of its existence, are a result of the high long-range order of the lattice $\cL$ in the CPS.

Next we provide an upperbound for the support of each of the three components in the generalized Eberlein decomposition. We prove in Theorem~\ref{T2} that, if $\oplam(W)$ is any closed model set containing the support of $\gamma$, then the generalized Eberlein decomposition of $\gamma$ stays within the class of measures supported inside $\oplam(W)$. This result generalizes some results of \cite{JBA,LS2,CRS3,NS1,NS5,ST}.
This type of decomposition was used to derive purity results for 1-dimensional Pisot substitution tilings (see \cite{BaGa,BG2}).

As an immediate consequence of Theorem~\ref{T2} we get the strong almost periodicity of each of the spectral components $\left(\widehat{\gamma}\right)_{pp} \,,\, \left(\widehat{\gamma}\right)_{ac}$ and $\left(\widehat{\gamma}\right)_{sc}$, respectively. In particular we obtain that each of the spectral components in the diffraction of Meyer sets is either trivial or has relatively dense support.

We continue our Fourier analysis by studying the norm-almost periodicity of the Fourier transform $\widehat{\gamma}$. We prove in Theorem~\ref{diff is NAP} that, if $\gamma$ is Fourier transformable and has Meyer set support, then $\widehat{\gamma}$ is a norm almost periodic measure. In particular we obtain that, under the same conditions, each of the measures $\left(\widehat{\gamma}\right)_{pp} \,,\, \left(\widehat{\gamma}\right)_{ac}$ and $\left(\widehat{\gamma}\right)_{sc}$, respectively, is norm almost periodic. Since norm almost periodicity is a stronger notion than strong almost periodicity, this improves our previous result, as well as some of the results of \cite{NS1,NS2,NS11}.

Next we study the connection between Fourier transformability of a measure with Meyer set support and the class of positive definite measures. It is well known that the space of measures spanned by positive definite measures is a subspace of the space of Fourier transformable measures \cite{ARMA1,BF,MoSt}. An open question posted by Argabright and deLamadrid is asking whether these two spaces are the same. So far, the question was answered only for measures with Lattice support \cite{CRS3}: in this case the two spaces are the same. We show in Theorem~\ref{T4} that the answer is also positive in the case of measures supported inside Meyer sets: a measure $\gamma$ supported inside a Meyer set $\Lambda$ is Fourier transformable if and only if it is a linear combination of positive definite measures. Moreover, we prove that there exist a larger Meyer set $\Gamma \supseteq \Lambda$, which depends on $\Lambda$ but can be chosen independent of $\gamma$, such that, $\gamma$ can be written as a linear combination of positive definite measures supported inside $\Gamma$.

All the results in this paper are then collected in Theorem~\ref{T3}.

\smallskip

We complete the introduction by presenting the general strategy for the proof of three main theorems: Theorem~\ref{T2}, Theorem~\ref{diff is NAP} and Theorem~\ref{T4}.

Consider any fixed Meyer set $\Lambda$. In Proposition~\ref{P1} we construct a twice Fourier transformable measure $\omega$ supported inside a larger Meyer set $\Gamma$ such that $\omega(\{ x \})=1$ for all $x \in \Lambda$ and $\widehat{\omega}$ is pure point. This is done by embedding the Meyer set into a model set, and by picking a nice function which $h$ in the internal space which is $1$ on the window of the model set.

Next, we  show that $\omega(\{ x \})=1$ for all $x \in \Lambda$, and that $\omega$ has uniformly discrete support implies that $\widehat{\omega}$ "dominates" $\widehat{\gamma}$ in the following sense: there exists a finite measure $\nu$ such that
\begin{equation}\label{EQ4}
\widehat{\gamma}=\widehat{\omega}*\nu \,.
\end{equation}
This phenomena is similar to the case of measures with lattice support \cite{CRS3}.

Therefore, using the uniqueness of the Lebesgue decomposition, we get
\begin{equation}\label{EQ2}
\left(\widehat{\gamma}\right)_{pp}=\widehat{\omega}*\nu_{pp} \quad ; \quad
\left(\widehat{\gamma}\right)_{ac}=\widehat{\omega}*\nu_{ac} \quad ; \quad
\left(\widehat{\gamma}\right)_{sc}=\widehat{\omega}*\nu_{sc} \,.
\end{equation}
We can use now the twice Fourier transformability of $\omega$ to show that each of the three measures in (\ref{EQ2}) has an inverse Fourier transform, which is supported inside $\Gamma=\supp(\omega)$, yielding the desired generalized Eberlein decomposition.

To prove norm almost periodicity, we show that the CPS and function $h$ used to produce the measure $\omega$ above can be chosen in such a way that $\widehat{\omega}$ is a norm almost periodic measure. As the convolution with finite measures preserves norm-almost periodicity (see Proposition~\ref{conv preserves NAP}) the conclusion of Theorem~\ref{diff is NAP} follows.

Finally, we use the polarisation identity \cite[Page~244]{MoSt} to show that $\omega$ can be written as a linear combination of four positive definite measures supported inside a common Meyer set $\Gamma$. Next, using Eq. (\ref{EQ4}) and writing $\nu$ as a linear combination of finite positive measures, we get via the inverse Fourier transform that
\begin{displaymath}
\gamma=\sum_{j=1}^4 C_j f_j \omega
\end{displaymath}
with $C_j \in \CC$ and $f_j \in \Cu(G)$ positive definite. Since $\omega$ is a linear combination of positive definite measures, and the product $f \nu$ of a positive definite function $f \in \Cu(G)$ and a positive definite measure $\nu$ is positive definite, we get Theorem~\ref{T4}.

\section{Definitions and Notations}

We start by briefly reviewing some basic definitions and properties. For a more detailed overview of these concepts we recommend the monographs \cite{TAO,TAO2}, as well as \cite{ARMA1,BF,ARMA,MoSt,NS11}.

Throughout the paper, $G$ denotes a second countable locally compact Abelian group (LCAG). By $\Cu(G)$ we denote the space of uniformly continuous and bounded functions on $G$. This is a Banach space with respect to the sup norm $\| \quad \|_\infty$. As usual, we denote by $C_0(G)$ the subspace of $\Cu(G)$ consisting of functions vanishing at infinity, and by $\Cc(G)$ the subspace of compactly supported continuous functions. Note that $\Cc(G)$ is not complete in $(\Cu(G), \| \quad \|_\infty)$.

As usual for diffraction theory, by measure we will understand a \textbf{Radon measure}. Via the Riesz Representation Theorem, a Radon measure is simply a linear functional on $\Cc(G)$, equipped with the inductive topology (see \cite[Appendix]{CRS2} for more details).

For a function $f$ on $G$ we denote by $f^\dagger$ its reflection, that is
\begin{displaymath}
f^\dagger(x):=f(-x) \; \mbox{ and } \; \tilde{f}(x)=\overline{f(-x)} \,.
\end{displaymath}
Same way, for a measure $\mu$ we denote by $\mu^\dagger$ the reflection $\mu^\dagger(f):= \mu (f^\dagger)$.

Finally, we denote by
\begin{displaymath}
K_2(G) := \mbox{Span} \{ f *g : f,g \in \Cc(G) \} \,.
\end{displaymath}

Let us recall here the definition of the Fourier transform of measures. For a more detailed review of the subject, we recommend \cite{MoSt}.

\begin{definition} A measure $\mu$ on $G$ is called \textbf{Fourier transformable} if there exists a measure $\widehat{\mu}$ on $\widehat{G}$ such that, for all $g \in K_2(G)$ we have
\begin{itemize}
\item[(i)] $\widecheck{g}  \in L^1( |\widehat{\mu} |)$,
\item[(ii)]$\displaystyle \langle \mu , g \rangle = \langle \widehat{\mu}, \widecheck{g} \rangle$.
\end{itemize}
\end{definition}

\smallskip
Next, let us review the concept of positive definite measures.

\begin{definition} A measure $\mu$ is called positive definite if and only if, for all $f \in \Cc(G)$ we have
\begin{displaymath}
\mu(f*\widetilde{f}) \geq  0\,.
\end{displaymath}
\end{definition}

It is easy to see that a measure $\mu$ is positive definite if and only if, for all $f \in \Cc(G)$ the function $\mu*f*\widetilde{f}$ is a continuous positive definite function \cite[Lemma~4.11.2]{MoSt}.

The importance of  positive definiteness to Fourier analysis is given by the following result.

\begin{theorem}\label{pd is FT}\cite{ARMA,MoSt,BF} Let $\mu$ be a positive definite measure. Then $\mu$ is Fourier transformable and its Fourier transform $\widehat{\mu}$ is positive.
\end{theorem}\qed

The converse of Theorem~\ref{pd is FT} also holds \cite{BF,MoSt}. For an review of positive definite measures see \cite{BF,MoSt}.

\smallskip
Next, let us recall translation boundedness for measures.

\begin{definition} A measure $\mu$ on $G$ is called \textbf{translation bounded} if for all compact sets $K \subseteq G$ we have
\begin{displaymath}
\| \mu \|_{K} := \sup_{t \in G} \left| \mu \right|(t+K) < \infty \,,
\end{displaymath}

We will denote by $\cM^\infty(G)$ the space of translation bounded measures on $G$.
\end{definition}

\begin{remark}
\begin{itemize}
  \item[(i)] A measure $\mu$ is translation bounded if and only if $\mu*f \in \Cu(G)$ for all $f \in \Cc(G)$ \cite[Thm.~1.1]{ARMA1}.
  \item[(ii)] $\| \quad \|_K$ is a norm on $\cM^\infty(G)$ \cite{BM} and $(\cM^\infty(G), \| \quad \|_K)$ is a Banach space \cite{CRS3}.
  \item[(iii)] If $K_1, K_2$ are two compact sets with non-empty interior, then the norms $\| \quad \|_{K_1}$ and $\| \quad \|_{K_2}$ are equivalent \cite{BM}. In particular, $\mu$ is translation bounded if and only if $\| \mu \|_K < \infty$ for one compact set $K$ with non-empty interior.
\end{itemize}
\end{remark}

\smallskip

We complete the section by reviewing almost periodicity for measures.

\begin{definition} A measure function $f \in \Cu(G)$ is called \textbf{Bohr almost periodic} if the hull $C_f := \{ T_tf | t \in G \}$ of all its translates has compact closure in $(\Cu(G), \| \quad \|_\infty)$.
We denote by $SAP(G)$ the subspace of $\Cu(G)$ of all Bohr almost periodic functions.

A measure $\mu$ is called \textbf{strong almost periodic} if for all $f \in \Cc(G)$ we have $\mu*f \in SAP(G)$. We denote the space of all strong almost periodic measures by $\SAP(G)$.

A measure $\mu$ is called \textbf{norm almost periodic} if for each $\epsilon>0$ the set
\begin{displaymath}
P_\epsilon^K(\mu):=\{ t \in G : \| T_t\mu-\mu \|_K < \epsilon \}
\end{displaymath}
of \textbf{$\epsilon$-norm almost periods} is relatively dense, where $K$ is any fixed compact set with non-empty interior.
\end{definition}

\begin{remark}
\begin{itemize}
  \item [(i)] Since different compact sets with non-empty interior define equivalent norms, norm almost periodicity is independent of the choice of the compact $K$.
  \item [(ii)] Norm almost periodic pure point measures were studied and classified in terms of CPS $(G, H, \cL)$ and continuous functions $h \in \Cz(H)$ in \cite{NS11}.
  \item [(iii)] A norm almost periodic measure $\mu$ is strong almost periodic. Moreover, if $\supp(\mu)-\supp(\mu)$ is uniformly discrete then $\mu$ is norm almost periodic if and only if $\mu$ is strong almost periodic \cite{BM}.
\end{itemize}
\end{remark}

It is easy to see that $\SAP(G) \subseteq \cM^\infty(G)$. The importance of the space $\SAP(G)$ in the study of the Lebesgue decomposition of Fourier transform is emphasized by the following two results.

\begin{theorem}\label{AM1}\cite[Cor.~11.1]{ARMA} Let $\gamma$ be a Fourier transformable measure. Then $\gamma$ is pure point if and only if $\widehat{\gamma} \in \SAP(G)$.
\end{theorem}\qed

\begin{theorem}\cite[Cor.~4.10.13]{MoSt} Let $\gamma$ be translation bounded Fourier transformable measure. Then $\widehat{\gamma}$ is pure point if and only if $\gamma \in \SAP(G)$.
\end{theorem}\qed

\bigskip

We complete this section with a brief review of Meyer sets. For more details, we recommend \cite{TAO,Meyer,Moody,MOO,NS11}.

\begin{definition}\label{CPS} By a \textbf{cut and project scheme} (or \textbf{CPS}) we understand a triple $(G, H, \cL)$, with $H$ a LCAG, and a lattice $\cL \subseteq G \times H$ such that
\begin{itemize}
  \item [(a)] $\pi_H(\cL)$ is dense in $H$.
  \item [(b)] the restriction $\pi_G|_\cL$ of the first projection $\pi_G$ to $\cL$ is one to one.
\end{itemize}
\end{definition}

\smallskip

Given a CPS $(G,H, \cL)$ we will denote by $L:= \pi_G(\cL)$. Then, $\pi_G$ induces a group isomorphism between $\cL$ and $L$. Composing the inverse of this with the second projection $\pi_H$ we get a mapping
\begin{displaymath}
\star : L \to H
\end{displaymath}
which we will call the \textbf{$\star$-mapping}. We can then write
\begin{displaymath}
\cL = \{ (x,x^\star) : x \in L \} \,.
\end{displaymath}

We can summarise a CPS in a simple diagram:

\begin{center}
\begin{tikzpicture}
  \matrix (m) [matrix of math nodes,row sep=1em,column sep=2em,minimum width=2em]
  {
     G & G \times H & H  \\
      \bigcup & \bigcup & \bigcup  \\
    L & \cL & L^\star \\};
  \path[-stealth]
    (m-1-2) edge node [above] {$\pi_G$}  (m-1-1)
            edge node [above] {$\pi_H$}  (m-1-3)
    (m-3-2) edge node [above] {1-1}  (m-3-1)
            edge node [above] {dense}  (m-3-3);
    \path[-stealth, bend right=20]
    (m-3-1.south) edge node [above] {$\star$}  (m-3-3.south);
\end{tikzpicture}
\end{center}

\medskip

Given a CPS $(G, H, \cL)$, we can define
\begin{displaymath}
\cL^0 := \{ (\chi, \psi) \in \widehat{G} \times \widehat{H} : \chi(x)\psi(x^\star) =1 \, \forall x \in L \} \,.
\end{displaymath}

Then, $(\widehat{G}, \widehat{H}, \cL^0)$ is a CPS \cite{BHS,MOO,Moody,NS11}. We will refer to this as the CPS \textbf{dual to $(G, H, \cL)$}.

\smallskip
We can now introduce the definition of a Meyer set.

\begin{definition} A Delone set $\Lambda \subset G$ is called a \textbf{Meyer set} if there exists a cut and project scheme $(G,H, \cL)$ and a compact set $W\subseteq H$ such that $\Lambda \subseteq \oplam(W)$.
\end{definition}

A Delone set is Meyer if and only if $\Lambda - \Lambda -\Lambda $ is closed and discrete, or equivalently, if and only if there exists a finite set $F$ such that $\Lambda- \Lambda \subseteq \Lambda +F$ \cite{LAG,LAG1,Meyer,MOO,NS11}.

\section{A Ping-Pong Lemma for measures with Meyer set support}

Next, we prove the following Lemma, which will play an important role in our approach.

\begin{lemma}\label{L2} Let $H$ be a LCAG, $W \subseteq H$ be a compact set and $y \in H$ be any point such that $y \notin W$. Then, there exists some $h \in K_2(H)$ such that $h \equiv 1$ on $W$ and $h(y)=0$.

In particular, there exists some $h \in K_2(H)$ such that $h \equiv 1$ on $W$.
\end{lemma}
\begin{proof}

Since $W$ is closed and $y \notin W$, there exists some open set $0 \in U=-U$ such that $(y+U) \cap W =\emptyset$. We can then pick some precompact open set $0 \in V=-V$ such that $V+V+V \subseteq U$.

As $0 \in V$ is open, we have $\overline{V} \subseteq V+V$. It follows that $W + \overline{V}$ is a compact subset of the open set $W+V+V$. Then, by the Urysohn's Lemma, there exists some $f \in \Cc(H)$ such that $1_{W + \overline{V}} \leq f \leq 1_{W+V+V}$.

Next, pick some $g \in \Cc(H)$ such that $g \geq 0, \supp(g) \subseteq V$ and $\int_H g(t) \dd t=1$. Set $h=f*g$. We claim that this $h$ has the desired properties.

By construction $f,g \in \Cc(H)$ and hence $h \in K_2(H)$. Since $f=1$ on $W + \overline{V}$ and $\supp(g) \subseteq V$, it is easy to see that $h(s)=\int_H g(t) \dd t =1$ for all $s \in W$.

Finally, as $\supp(f) \subseteq W+V+V$ and $\supp(g) \subseteq V$ we have $\supp(h) \subseteq W+V+V+V \subseteq W+U=W-U$. Therefore $h(y)=0$ as $y \notin W-U$.

The last claim is immediate. Indeed, if $W \neq H$ then the claim follows from the above by picking some $y \notin W$. Otherwise, $H$ is compact and $h=1_{H}$ works.
\end{proof}

As a consequence we get:

\begin{proposition}\label{P1} Let $G$ be a second countable group, and $\Lambda \subset G$ be a Meyer set. Let $(G, H, \cL)$ be any cut and project scheme and $W \subseteq H$ compact be such that $\Lambda \subseteq \oplam(W)$.

Then, there exists some $h \in K_2(H)$ such that, the measure
\begin{displaymath}
\omega_h:= \sum_{(x,x^\star) \in \cL} h(x^\star) \delta_x
\end{displaymath}
has the following three properties:

\begin{itemize}
\item[(i)] $\omega_h(\{ x\})=1 \forall x \in \Lambda$.
\item[(ii)] $\omega_h$ is twice Fourier transformable and
\begin{displaymath}
 \widehat{\omega_{h}}= \det(\cL) \omega_{\widecheck{h}}= \det(\cL) \sum_{(y,y^\star) \in \cL^0} \widecheck{h}(y^\star) \delta_y \,.
\end{displaymath}
\item[(iii)] $\Gamma:=\supp(\omega_h)$ is a Meyer set.
\end{itemize}

Moreover, for each $z \in G \backslash \oplam(W)$, we can pick $h$ in such a way that $z \notin \Gamma=\supp(\omega_h)$.
\end{proposition}
\begin{proof}

By Lemma~\ref{L2}, there exists some $h \in K_2(H)$, such that $h =1$ on $W$.

Since $h \in K_2(H)$, the measure $\omega_h$ is twice Fourier transformable \cite[Thm.~4.10, Thm.~4.12]{CRS2} and
\begin{displaymath}
\widehat{\omega_{h}}= \det(\cL) \omega_{\widecheck{h}}= \det(\cL) \sum_{(y,y^\star) \in \cL^0} \widecheck{h}(y^\star) \delta_y \,.
\end{displaymath}

Next, the set $U:= \{ x \in H | h(x) \neq 0 \}$ is open and precompact, and hence $\Gamma:= \oplam(U)=\supp(\omega_h)$ is a Meyer set.

It follows that $\omega_h$ satisfies properties (i), (ii) and (ii).

Finally, let  $z \in G \backslash \oplam(W)$ be arbitrary. If $z \notin \pi_G( \cL)$, then we have $z \notin \Gamma=\oplam(U)$ for all choices of $U$.

If $z \in \pi_G( \cL)$, then $z \notin \oplam(W)$ implies that $z^\star \notin W$. Then, by Lemma~\ref{L2}, we can pick $h \in K_2(H)$ such that $h(z^\star)=0$, and hence, by the above construction, $z \notin \Gamma$.
\end{proof}

\bigskip

Recall that given a Fourier transformable measure $\mu$ supported inside a lattice $L \subset G$, its Fourier transform $\widehat{\mu}$ is fully periodic under the lattice $L^0 \subset \widehat{G}$ dual to $L$ \cite{ARMA1,BF}. It follows \cite{CRS3} that there exists a finite measure $\nu$ such that
\begin{displaymath}
\widehat{\mu}=\delta_{L^0} * \nu \,.
\end{displaymath}

Below we prove that a similar result holds for Fourier transformable measures with Meyer set support. Starting with a measure $\gamma$ with Meyer support, this result will allow us to move to the Fourier space, take certain decompositions and return to the original space.

\begin{lemma}[Ping-Pong Lemma for Meyer sets]\label{L1} Let $\Lambda \subset G$ be a Meyer set, $\omega$ a twice Fourier transformable measure with the following two properties:
\begin{itemize}
\item[(a)] $\Gamma:= \supp( \omega)$ is $U$-uniformly discrete for some open set $0 \in U$;
\item[(b)] $\omega(\{x \})=1$ for all $x \in \Lambda$.
\end{itemize}
Let $f \in \Cc(G)$ be any function such that $\supp(f*\widetilde{f}) \subseteq U$.

Then, the following hold:

\begin{itemize}
\item[(i)] If $\gamma$ is any Fourier transformable measure with $\supp(\gamma) \subseteq \Lambda$, then $\nu:= \left| \widecheck{f} \right|^2 \widehat{\gamma}$ is a finite measure and
    \begin{displaymath}
\widehat{\gamma}=\widehat{\omega}* \nu \,.
    \end{displaymath}
\item[(ii)] If $\nu$ is a finite measure on $\widehat{G}$, then $\widehat{\omega}*\nu$ is Fourier transformable and $\gamma:= \left( \reallywidehat{\widehat{\omega}*\nu} \right)^\dagger$ is supported inside $\Gamma$. Moreover, if $\gamma$ is Fourier transformable, then
\begin{displaymath}
\widehat{\gamma}=\widehat{\omega}* \nu \,.
\end{displaymath}
\end{itemize}
\end{lemma}
\begin{proof}

\textbf{(i)}

Since $\gamma$ is Fourier transformable, for each $f \in \Cc(G)$ the measure $\nu:=\widehat{\gamma} \left| \widecheck{f} \right|^2$ is finite, and has Fourier transform as finite measure \cite[Prop~2.3]{ARMA1}, \cite[Lemma~4.9.24]{MoSt} given by
\begin{displaymath}
\widehat{\nu}=\reallywidehat{\widehat{\gamma} \left| \widecheck{f} \right|^2}= \left( \gamma*f*\widetilde{f} \right)^\dagger \,.
\end{displaymath}

Next, since $\widehat{\omega}$ is translation bounded (see \cite[Thm.~2.5]{ARMA1} or \cite[Thm.~4.9.23]{MoSt}) $\widehat{\gamma} \left| \widecheck{f} \right|^2$ is finite, the measures $\widehat{\omega}$ and $\nu$ are convolvable \cite[Thm.~1.2]{ARMA1}, \cite[Lemma 4.9.19]{MoSt}. Moreover, as $\widehat{\omega}$ is Fourier transformable and $\nu$ is finite, their convolution is Fourier transformable and \cite[Lemma 4.9.26]{MoSt}
\begin{equation}\label{EQ1}
\reallywidehat{\widehat{\omega} *\nu}=\widehat{\widehat{\omega}}\widehat{\nu}=\widehat{\widehat{\omega}}\reallywidehat{\left(\widehat{\gamma} \left| \widecheck{f} \right|^2\right)}=\left( \gamma*f*\widetilde{f} \right)^\dagger \omega^\dagger= \left(\left( \gamma*f*\widetilde{f} \right) \omega \right)^\dagger \,.
\end{equation}

We claim that
\begin{displaymath}
\left( \gamma*f*\widetilde{f} \right) \omega =\gamma \,.
\end{displaymath}

Indeed,
\begin{displaymath}
\supp\bigl(\left( \gamma*f*\widetilde{f} \right) \omega \bigr) \subseteq \supp(\gamma*f*\widetilde{f} ) \cap \supp(\omega ) \subseteq (\Lambda +U ) \cap \Gamma =\Lambda \,.
\end{displaymath}
Moreover, for each $x \in \Lambda$, since $\omega(\{x\})=1$ and $\Gamma$ is $U$-uniformly discrete, we have
\begin{displaymath}
\bigl(\left( \gamma*f*\widetilde{f} \right) \omega \bigr) (\{ x\})= \gamma*f*\widetilde{f}(\{ x \}) = \sum_{y \in \Lambda} \left(f*\widetilde{f}\right)(x-y) \gamma(\{ y \}) \,.
\end{displaymath}

The last sum is zero unless $x-y \in U$ and $y \in \Lambda$, which means $y \in \Lambda \cap (x-U)= \{ x \}$. Therefore
\begin{displaymath}
\left( \gamma*f*\widetilde{f} \right) \omega (\{ x\})= \left(f*\widetilde{f}\right)(0) \gamma(\{ x \})= \gamma(\{ x \}) \,.
\end{displaymath}

Now, since $\gamma$ is Fourier transformable, so is $\gamma^\dagger= \reallywidehat{\widehat{\omega} * \left(\widehat{\gamma} \left| \widecheck{f} \right|^2\right)}$. Therefore, by applying \cite[Thm.~4.9.28]{MoSt} to (\ref{EQ1}) we have
\begin{displaymath}
\reallywidehat{\gamma^\dagger}=\reallywidehat{\reallywidehat{\widehat{\omega} * \left(\widehat{\gamma} \left| \widecheck{f} \right|^2\right)}}= \left( \widehat{\omega} * \left(\widehat{\gamma} \left| \widecheck{f} \right|^2 \right)\right)^\dagger \,.
\end{displaymath}

Reflecting this relation we get the claim.

\textbf{(ii)} Exactly as before, since $\omega$ is twice Fourier transformable, $\widehat{\omega}$ is translation bounded and Fourier transformable. Then, as $\nu$ is finite, $\widehat{\omega}*\nu$ is well defined, Fourier transformable, and
\begin{displaymath}
\gamma=\left(\reallywidehat{\widehat{\omega}*\nu} \right)^\dagger = \left( \widehat{\widehat{\omega}} \widehat{\nu} \right)^\dagger= \left(\widehat{\nu} \right)^\dagger \omega \,.
\end{displaymath}

Since $\nu$ is finite, we have $\widehat{\nu} \in \Cu(G)$, and hence

\begin{displaymath}
\supp(\gamma) \subseteq \supp(\omega) \subseteq \Gamma \,.
\end{displaymath}

If $\gamma$ is Fourier transformable, it follows that $\widehat{\omega}*\nu$ is twice Fourier transformable. Then, by \cite[Thm.~3.4]{ARMA}, \cite[Thm.~4.9.28]{MoSt} we get
\begin{displaymath}
\reallywidehat{\reallywidehat{\left(\widehat{\omega}*\nu\right)^\dagger}}= \widehat{\omega}*\nu \,.
\end{displaymath}

This gives $\widehat{\gamma}= \widehat{\omega}*\nu$.

\end{proof}

\begin{remark} In the theory of mathematical diffraction the measures $\widehat{\gamma}$ will always be positive. It follows that all measures $\gamma$ appearing in Lemma~\ref{L1} (ii) will be positive definite, and hence Fourier transformable.
\end{remark}

\begin{remark} If $\Lambda=L$ is a lattice, then one can take $\omega=\delta_{L}$, and hence $\Gamma=L$. Moreover, any $L^0$-periodic measure is a Fourier transform \cite{CRS3}.

In this case, since $\Gamma=L=\Lambda$, Lemma~\ref{L1} gives the well known result (compare \cite{BA,CRS3}) that a Fourier transformable measure $\gamma$ is supported inside $L$ if and only if there exists a finite measure $\nu$ such that
\begin{displaymath}
\widehat{\gamma}= \delta_{L^0} * \nu \,.
\end{displaymath}
\end{remark}

\begin{remark} Let $L \subset G$ be a lattice, let $F \subset G$ be a finite set, and set $\Lambda=L+F$.

Now, if $F'$ is a minimal subset of $F$ with the property that $\Lambda=L+F'$, then we have $\delta_\Lambda=\delta_{F'}*\delta_L$. We can then pick $\omega:= \delta_{\Lambda}$, and hence $\Gamma=\Lambda$.

Lemma~\ref{L1} then yields that a Fourier transformable measure $\gamma$ is supported inside $\Lambda$ if and only if there exists a finite measure $\nu$ such that
\begin{displaymath}
\widehat{\gamma}= \left(\sum_{\chi \in L^0} \bigl(\sum_{f \in F'} \chi(f) \bigr) \delta_\chi \right)  * \nu \,.
\end{displaymath}
\end{remark}

\section{On the Generalized Eberlein decomposition for measures with Meyer set support}

We can now prove the existence of the generalized Eberlein decomposition for Fourier transformable measures with Meyer set support.

\begin{theorem}[\textbf{Existence of the generalized Eberlein decomposition}]\label{T2} Let $G$ be a second countable LCAG, and let $\Lambda \subset G$ be a Meyer set.

Then, for each Fourier transformable measure $\gamma$ with $\supp(\gamma) \subseteq \Lambda$, there exist unique Fourier transformable measures $\gamma_{s}, \gamma_{0s}, \gamma_{0a}$ supported inside a Meyer set, such that
\begin{align*}
\gamma=\gamma_{s} &+ \gamma_{0s} + \gamma_{0a} \\
\reallywidehat{\gamma_s}&= \left(\widehat{\gamma}\right)_{pp}  \\
\reallywidehat{\gamma_{0s}} &=  \left(\widehat{\gamma}\right)_{sc} \\
\reallywidehat{\gamma_{0a}} &=\left(\widehat{\gamma}\right)_{ac} \,.  \\
\end{align*}

Moreover, if $(G,H, \cL)$ is any CPS and $W \subseteq H$ any compact set such that $\Lambda \subseteq \oplam(W)$, then
\begin{displaymath}
\supp(\gamma_{s}), \supp(\gamma_{0s}), \supp(\gamma_{0a}) \subseteq \oplam(W) \,.
\end{displaymath}
\end{theorem}
\begin{proof}

Now, under the notations from Lemma~\ref{L1}, by Lemma~\ref{L1} there exists a finite measure $\nu$ such that
\begin{displaymath}
\widehat{\gamma}=\widehat{\omega}*\nu \,.
\end{displaymath}

Since $G$ is metrisable, $\widehat{G}$ is $\sigma$ compact. Therefore, the pure point measure $\widehat{\omega}$ has (at most) countable support.

As $\widehat{G}$ is $\sigma$-compact, the measure $\nu$ admits a Lebesgue decomposition
\begin{displaymath}
\nu=\nu_{pp}+\nu_{ac}+\nu_{sc} \,.
\end{displaymath}

Since $\widehat{\omega}$ is pure point measure, the measures $\widehat{\omega}*\nu_{pp}, \widehat{\omega}*\nu_{ac}, \widehat{\omega}*\nu_{sc}$ are pure point, absolutely continuous and singular continuous, respectively. Moreover, we have
\begin{displaymath}
\widehat{\gamma}=\widehat{\omega}*\nu=\widehat{\omega}*\nu_{pp}+\widehat{\omega}*\nu_{ac} +\widehat{\omega}*\nu_{sc}
\end{displaymath}
and therefore, by the uniqueness of the Lebesgue decomposition, we have
\begin{displaymath}
\left(\widehat{\gamma}\right)_{pp}=\widehat{\omega}*\nu_{pp} \quad ; \quad
\left(\widehat{\gamma}\right)_{ac}=\widehat{\omega}*\nu_{ac} \quad ; \quad
\left(\widehat{\gamma}\right)_{sc}=\widehat{\omega}*\nu_{sc} \,.
\end{displaymath}

By Lemma~\ref{L1} (ii), the measures $\widehat{\omega}*\nu_{pp}, \widehat{\omega}*\nu_{ac}, \widehat{\omega}*\nu_{sc}$ are Fourier transformable, and the measures
\begin{align*}
\gamma_{s}&:=\left( \reallywidehat{\widehat{\omega}*\nu_{pp}} \right)^\dagger \\
\gamma_{0a}&:=\left( \reallywidehat{\widehat{\omega}*\nu_{ac}} \right)^\dagger \\
\gamma_{0s}&:=\left( \reallywidehat{\widehat{\omega}*\nu_{sc}} \right)^\dagger \\
\end{align*}
are supported inside $\Gamma$.

 We next show that $\gamma_{s}, \gamma_{0s}$ and $\gamma_{0a}$ are Fourier transformable, or equivalently that $\widehat{\omega}*\nu_{pp} , \widehat{\omega}*\nu_{ac}, \widehat{\omega}*\nu_{sc}$ are twice Fourier transformable. By \cite[Thm.~3.10]{CRS} this is equivalent to the integrability of $\widecheck{g}$ with respect to each of these measures, for all $g \in K_2(G)$. We show below that this is an immediate consequence of the Fourier transformability of $\gamma$.

\smallskip

Since $\gamma$ is Fourier transformable, for each $g \in K_2(G)$ we have $\widecheck{g} \in L^1( \widehat{\gamma}^\dagger)$, and hence  $\widecheck{g} \in L^1( \left(\widehat{\gamma}\right)_{pp}^\dagger) , \widecheck{g} \in L^1( \left(\widehat{\gamma}\right)_{sc}^\dagger)$ and $\widecheck{g} \in L^1( \left(\widehat{\gamma}\right)_{ac}^\dagger)$. Therefore, since $\left(\widehat{\gamma}\right)_{pp}=\widehat{\omega}*\nu_{pp},
\left(\widehat{\gamma}\right)_{ac}=\widehat{\omega}*\nu_{ac}$ and $\left(\widehat{\gamma}\right)_{sc}=\widehat{\omega}*\nu_{sc}$ are Fourier transformable, by \cite[Thm.~3.10]{CRS}  we get that $\widehat{\omega}*\nu_{pp} , \widehat{\omega}*\nu_{ac}, \widehat{\omega}*\nu_{sc}$ are twice Fourier transformable.

This proves the first claim.

\smallskip

Let now $(G,H, \cL)$ be any CPS and $W \subseteq H$ be any compact set such that $\Lambda \subseteq \oplam(W)$. We show that
\begin{displaymath}
\supp(\gamma_{s}), \supp(\gamma_{0s}), \supp(\gamma_{0a}) \subseteq \oplam(W) \,.
\end{displaymath}

Indeed, for each $z \in G \backslash \oplam(W)$, by Proposition~\ref{P1} we can pick $\omega$ in Lemma~\ref{L1} such that $\omega(\{ z \})=0$.

Then, in the above we have $z \notin \Gamma$ and hence
\begin{displaymath}
\gamma_{s}(\{z \})= \gamma_{0s}(\{z \})=\gamma_{0a} (\{z \})= 0 \,.
\end{displaymath}

This proves the claim.
\end{proof}

\begin{definition} We say that a measure $\gamma$ admits a \textbf{generalized Eberlein decomposition} if we can find Fourier transformable measures $\gamma_{s}, \gamma_{0s}, \gamma_{0a}$ such that
\begin{align*}
\gamma=\gamma_{s} &+ \gamma_{0s} + \gamma_{0a} \\
\reallywidehat{\gamma_s}&= \left(\widehat{\gamma}\right)_{pp}  \\
\reallywidehat{\gamma_{0s}} &=  \left(\widehat{\gamma}\right)_{sc} \\
\reallywidehat{\gamma_{0a}} &=\left(\widehat{\gamma}\right)_{ac} \,.  \\
\end{align*}
\end{definition}

The injectivity of the Fourier transform gives that, whenever the generalized Eberlein decomposition exists, it is unique. In Theorem~\ref{T2} above we showed that Fourier transformable measures with Meyer set support always admit generalized Eberlein decomposition.

\smallskip

We complete the section by listing an interesting consequence of Theorem~\ref{T2}.

\begin{corollary}\label{CT2} Let $G$ be a second countable LCAG, and let $\Lambda \subset G$ be a Meyer set.

If $\gamma$ is a Fourier transformable measure with $\supp(\gamma) \subseteq \Lambda$, then
\begin{displaymath}
\left(\widehat{\gamma}\right)_{pp} , \left(\widehat{\gamma}\right)_{sc} , \left(\widehat{\gamma}\right)_{ac} \in \SAP(G) \,.
\end{displaymath}
\end{corollary}
\begin{proof} By Theorem~\ref{T2} , $\gamma$ admits a generalized Eberlein decomposition. Since $\gamma_{s}, \gamma_{0s},\gamma_{0a}$ are pure point, their Fourier transforms are strongly almost periodic measures by Theorem~\ref{AM1}.
\end{proof}

\section{On the norm-almost periodicity of a class of measures}

In the remaining of this paper we prove that given a Fourier transformable measure $\gamma$ with Meyer set support, then its Fourier transform $\widehat{\gamma}$ is a norm almost periodic measure.

We start by proving that, given a cut and project scheme of the form $(G, \RR^d \times H, \cL)$ for LCAG $H$ and some function $h=\phi \otimes \psi \in \Cc(\RR^d \times H)$ where $\phi \in \cS(\RR^d) \cap K_2(\RR^d), \psi \in \Cc(H)$ then, $\omega_h$ is  norm-almost periodic. Here, the product $\phi \otimes \psi $ is defined via
\begin{displaymath}
\phi \otimes \psi(x,y) := \phi(x) \psi(y) \,.
\end{displaymath}

We will complete the paper by showing that, given a Meyer set, we can find a measure $\omega$ satisfying the conditions in Lemma~\ref{L1} such that, by the results in this section, $\widehat{\omega}$ is norm-almost periodic and that norm-almost periodicity is preserved by convolution with finite measures.

\bigskip

First let us recall a result of \cite{NS11} which we need in this section.

\begin{proposition}\label{prop Ap1}\cite[Thm.~5.5.2]{NS11} Let $(G, H, \cL)$ be a CPS and $h \in \Cc(H)$. Then, $\omega_h$ is norm almost periodic.
\end{proposition}

In the Lemma~\ref{lemma AP1} below, we show that for CPS with $H=\RR^d \times H_1$ and $h= \phi \otimes \psi$ for some Schwartz function $\phi \in \cS(\RR^d)$ and some $\psi \in \Cc(H_1)$ then $\omega_h$ is a measure, and we give an upperbound for its norm $\| \omega_h \|_K$

\begin{lemma}\label{lemma AP1} Let $(G, \RR^d \times H_1, \cL)$ be a CPS. Then,
\begin{itemize}
\item[(i)] Let $\phi \in \cS(\RR^d), \psi \in \Cc(H_1)$ and let $h=\phi \otimes \psi$. Then, $\omega_h$ is a measure.
\item[(ii)] For each compact $K \subseteq G$ and compact $W \subseteq H_1$, there exists some constant $C=C(K,W)$ such that, for all $\phi \in \cS(\RR^d), \psi \in \Cc(H_1)$ with $\supp(\psi) \subseteq W$ we have
\begin{displaymath}
\| \omega_{\phi \otimes \psi} \|_K \leq C \|\psi\|_\infty \bigl\| (1+x^{2d}) \phi(x) \bigr\|_\infty \,.
\end{displaymath}
\end{itemize}
\end{lemma}
\begin{proof}

Let $W \subseteq H_1$ be any compact set, and let $\psi \in \Cc(H)$ be any function such that $\supp(\psi) \subseteq W$. Set

\begin{displaymath}
C(W,K):=\left( \| \delta_{\cL} \|_{K \times [-\frac{1}{2}, \frac{1}{2}]^d \times W} \right) \left( \sum_{n \in \ZZ^d}  \sup_{z \in n+[-\frac{1}{2},\frac{1}{2}]^d } \frac{1}{1+|z|^{2d}} \right)  < \infty \,.
\end{displaymath}

To prove that $\omega_h$ is a measure, we show that for each compact set $K$ we have
\begin{displaymath}
\sum_{x \in K} \left| \omega_h(\{ x \}) \right| < \infty \,.
\end{displaymath}

A simple computation yields
\begin{align*}
\sum_{x \in K} \left| \omega_h(\{ x \}) \right| =\sum_{(x,x^\star) \in \cL} \left| 1_{K}(x) h(x^\star) \right| =\sum_{(x,x^\star) \in \cL} \left| 1_{K}(x) \left(\phi \otimes \psi\right) (x^\star) \right| \,.
\end{align*}

Now, consider the canonical projections $\pi_{\RR^d}: \RR^d \times H_1 \to \RR^d$ and $\pi_{H_1}: \RR^d \times H_1 \to H_1$ , respectively.
For each $(x,x^\star) \in \cL$ we will denote, for simplicity, by $x_1^{\star}:=\pi_{\RR^d}(x^\star)$ and $x_2^{\star}:=\pi_{H_1}(x^\star)$.

Then
\label{EQ norm comp}
\begin{align}
&\sum_{x \in K} \left| \omega_h(\{ x \}) \right| =\sum_{(x,x^\star) \in \cL} \left| 1_{K}(x) \phi(x_1^\star) \psi (x_2^\star) \right| \\
&=\sum_{(x,x^\star) \in \cL}  1_{K}(x) \left| \phi(x_1^\star)\right| \left| \psi (x_2^\star) \right| \leq \sum_{(x,x^\star) \in \cL}  1_{K}(x) \left| \phi(x_1^\star)\right| \| \psi\|_\infty 1_{W} (x_2^\star) \nonumber\\
&\leq \sum_{ m \in \ZZ^d } \bigl( \sum_{\substack{(x,x^\star) \in \cL \\ x_1^\star \in m+[-\frac{1}{2} , \frac{1}{2}]^d }}  1_{K}(x) \left| \phi(x_1^\star)\right|  \| \psi\|_\infty 1_{W} (x_2^\star) \bigr) \nonumber\\
&= \| \psi\|_\infty  \sum_{ m \in \ZZ^d } \bigl( \sum_{(x,x^\star) \in \cL}  1_{K}(x) 1_{m+[-\frac{1}{2} , \frac{1}{2}]^d }(x_1^\star) \left| \phi(x_1^\star)\right|  1_{W} (x_2^\star) \bigr) \nonumber\\
&= \| \psi\|_\infty  \sum_{ m \in \ZZ^d } \bigl( \sum_{(x,x\star) \in \cL}  1_{K}(x) 1_{m+[-\frac{1}{2} , \frac{1}{2}]^d }(x_1^\star) \left| \phi(x_1^\star)(1+|x_1^\star|^{2d})\right| \left| \frac{1}{1+|x_1^\star|^{2d}} \right| 1_{W} (x_2^\star) \bigr) \nonumber\\
&\leq \| \psi\|_\infty  \sum_{ m \in \ZZ^d } \bigl( \sum_{(x,x^\star)\in \cL}  1_{K}(x) 1_{m+[-\frac{1}{2} , \frac{1}{2}]^d }(x_1^\star) \| \phi(s)(1+|s|^{2d}) \|_\infty \left( \sup_{z \in m+[-\frac{1}{2},\frac{1}{2}]^d } \frac{1}{1+|z|^{2d}} \right) 1_{W} (x_2^\star) \bigr) \nonumber\\
&= \| \psi\|_\infty  \| \phi(s)(1+|s|^{2d}) \|_\infty \sum_{ m \in \ZZ^d } \bigl( \sum_{(x,x^\star) \in \cL}  1_{K}(x) 1_{m+[-\frac{1}{2} , \frac{1}{2}]^d }(x_1^\star)  \left( \sup_{z \in m+[-\frac{1}{2},\frac{1}{2}]^d } \frac{1}{1+|z|^{2d}} \right) 1_{W} (x_2^\star) \bigr) \nonumber\\
&= \| \psi\|_\infty  \| \phi(s)(1+|s|^{2d})   \|_\infty \left(\sum_{ m \in \ZZ^d }  \sup_{z \in m+[-\frac{1}{2},\frac{1}{2}]^d } \frac{1}{1+|z|^{2d}}  \bigr) \bigl( \sum_{(x,x_1^\star, x_2^\star) \in \cL}  1_{K}(x) 1_{m+[-\frac{1}{2} , \frac{1}{2}]^d } (x_1^\star) 1_{W} (x_2^\star) \right) \nonumber\\
&\leq \left(  \sum_{ m \in \ZZ^d }  \sup_{z \in m+[-\frac{1}{2},\frac{1}{2}]^d } \frac{1}{1+|z|^{2d}}  \right)  \bigl( \delta_{\cL}  ( (t+K) \times  (m+[-\frac{1}{2} , \frac{1}{2}]^d ) \times W )\bigr) \| \psi\|_\infty  \| \phi(s)(1+|s|^d) \|_\infty  \nonumber\\
\nonumber
\end{align}

This shows that
\begin{displaymath}
\sum_{x \in K} \left| \omega_h(\{ x \}) \right| < \infty \,,
\end{displaymath}
and proves \textbf{(i)}.

We now show \textbf{(ii)}

For each compact set $K$ we have
\begin{align*}
  &\| \omega_{h} \|_K  = \sup_{t \in G} \left| \omega_h \right| (t+K) \\
  &\stackrel{(\ref{EQ norm comp})}{\leq}  \| \psi\|_\infty  \| \phi(s)(1+|s|^d)\|_\infty \left( \sum_{ m \in \ZZ^d } \bigl( \sup_{z \in m+[-\frac{1}{2},\frac{1}{2}]^d } \frac{1}{1+|z|^{2d}}  \bigr)  \bigl( \delta_{\cL}  ( (t+K) \times  (m+[-\frac{1}{2} , \frac{1}{2}]^d ) \times W )\bigr)\right) \\
   & \leq  \bigl(  \sum_{ m \in \ZZ^d }  \sup_{z \in m+[-\frac{1}{2},\frac{1}{2}]^d } \frac{1}{1+|z|^{2d}}  \bigr)  \| \psi\|_\infty  \| \phi(s)(1+|s|^d)\|_\infty  \| \delta_{\cL} \|_{  K \times  [-\frac{1}{2} , \frac{1}{2}]^d  \times W } \\
   &\leq C(K,W)  \| \psi\|_\infty  \| \phi(s)(1+|s|^d) \|_\infty \,.
\end{align*}

This completes the proof.

\end{proof}

\begin{lemma}\label{lemma AP2} Let $(G, \RR^d \times H_1, \cL)$ be a CPS. Let $\phi \in \cS(\RR^d), \psi \in \Cc(H_1)$ and set $h=\phi \otimes \psi$. Then $\omega_h$ is norm almost periodic.

\end{lemma}
\begin{proof}

Fix compact sets $K \subseteq G$ and $W \subseteq H$ such that $\supp(\psi) \subseteq W$.

Let $\epsilon >0$, and let $C=C(K,W)$ be the constant from Lemma~\ref{lemma AP1} (ii).

Since $\phi \in \cS(\RR^d)$ we can write it as $\phi=\phi_1+\phi_2$, with $\phi_1 \in \Cc(\RR^d)$ and
\begin{displaymath}
\| (1+|s|^{2d}) \phi_2(s) \|_\infty \leq \frac{\epsilon}{1+3 C \|\psi\|_\infty \| \delta_{\cL} \|_{K_1 \times [-\frac{1}{2},\frac{1}{2}]^d \times W} } \,.
\end{displaymath}

Then, by Lemma~\ref{lemma AP1} we have
\begin{displaymath}
\| \omega_{ \phi_2 \otimes \psi} \|_{K} < \frac{\epsilon}{3} \,.
\end{displaymath}

Next, by Proposition~\ref{prop Ap1} the measure $\omega_{\phi_1 \otimes \psi}$ is norm almost periodic. Therefore, the set
\begin{displaymath}
P:= \{ t \in G | \| T_t \omega_{\phi_1 \otimes \psi} - \omega_{\phi_1 \otimes \psi} \|_K < \frac{\epsilon}{3} \}
\end{displaymath}
is relatively dense.

Let $t \in P$. Then
\begin{align*}
 \| T_t \omega_{h} - \omega_{h} \|_K & =\| T_t \omega_{\phi_1 \otimes \psi} +T_t \omega_{\phi_2 \otimes \psi}- \omega_{\phi_1 \otimes \psi}- \omega_{\phi_2 \otimes \psi} \|_K \\
  & \leq \| T_t \omega_{\phi_1 \otimes \psi} - \omega_{\phi_1 \otimes \psi}\|_K +\| T_t \omega_{\phi_2 \otimes \psi}\|_K+| \omega_{\phi_2 \otimes \psi} \|_K \\
  & \leq \| T_t \omega_{\phi_1 \otimes \psi} - \omega_{\phi_1 \otimes \psi}\|_K +\| \omega_{\phi_2 \otimes \psi}\|_K+| \omega_{\phi_2 \otimes \psi} \|_K <\epsilon
\end{align*}

This shows that each $t \in P$ is an $\epsilon$-norm almost period for $\omega_h$, from which our claim follows.
\end{proof}

\section{Norm almost periodicity and convolution}

In this section we show that norm almost periodicity is preserved under convolution with finite measures. We start by proving the following preparation Lemma.

\begin{lemma}\label{conv finite estimate} Let $K, K'$ be compact sets such that $K \subset \mbox{Int} (K')$.
If $\nu$ is a finite measure and $\mu$ is any translation bounded measure, then
\begin{displaymath}
\left| \mu* \nu \right|(K) \leq   \| \mu \|_{K'} \left( \left| \nu \right|(G)\right) \,.
\end{displaymath}
\end{lemma}
\begin{proof}

Since $K \subset \mbox{Int} (K')$, there exists a non-negative $f$ such that $1_{K} \leq f \leq 1_{K'}$.

Let $\epsilon >0$ be arbitrary.

By the definition of the total variation, there exists some $g \in \Cc(G)$ with $|g| \leq f$ such that
\begin{displaymath}
| \mu * \nu | (f)  \leq \left| \mu*\nu (g) \right| +\epsilon \,.
\end{displaymath}

Then, we have
\begin{align*}
&\left| \mu* \nu \right|(K)  \leq | \mu * \nu | (f)  \leq \left| \mu*\nu (g) \right| +\epsilon = \left|  \int_{G} \int_{G} g(s+t) \dd \nu(s) \dd \mu(t)\right| +\epsilon \\
&\leq \int_{G}   \int_{G}  \left| g(s+t)\right| \dd  \left| \nu\right|(s) \dd  \left| \mu\right|(t) +\epsilon = \int_{G} \left(   \int_{G}  \left| g(s+t)\right| \dd  \left| \mu\right|(t) \right) | \dd  \left| \nu\right|(s) +\epsilon \\
\end{align*}

Now, for each $s \in G$ we have $g(s+t)=0$ for all $t \notin -s+K'$. Therefore
\begin{align*}
& \int_{G}  \left| g(s+t)\right| \dd  \left| \mu\right|(t) = \int_{-s+K'}  \left| g(s+t)\right| \dd  \left| \mu\right|(t)  \\
   & \leq \int_{-s+K'}  1 \dd  \left| \mu\right|(t)  = \left| \mu \right| (-s+K') \leq \| \mu \|_{K'}  \,. \\
\end{align*}

Therefore
\begin{displaymath}
\left| \mu* \nu \right|(K) \leq  \int_{G} \| \mu \|_{K'}  \, \dd  \left| \nu\right|(s) +\epsilon = \| \mu \|_{K'}  \left( \left| \nu \right|(G)\right) +\epsilon \,.
\end{displaymath}

As $\epsilon >0$ was arbitrary, we are done.

\end{proof}

We can now prove the desired result.

\begin{proposition}\label{conv preserves NAP} If $\nu$ is a finite measure and $\mu$  is a norm almost periodic, then $\mu*\nu$ is norm almost periodic.
\end{proposition}
\begin{proof}

Fix compact set $K_1 \subset G$ with non-empty interior. Pick $K_1'$ any compact set such that $K_1 \subseteq \mbox{Int}(K_1')$ .

Let $\epsilon >0$. Since $\mu$ is norm almost periodic, the set
\begin{displaymath}
P:= \{ t \in G | \| T_t \mu - \mu \|_{K_1} < \frac{\epsilon}{1+2  \left( \left| \nu \right|(G)\right)} \}
\end{displaymath}
is relatively dense. Then, for all $t \in P$ and $s \in G$ we have
\begin{displaymath}
\left| T_t \mu* \nu - \mu* \nu \right|(s+K_1) = \left| (T_t \mu-\mu) * \nu \right|(s+K_1)
\end{displaymath}

Now, since $s+K_1 \subseteq \mbox{Int}(s+K_1')$, by applying Lemma~\ref{conv finite estimate} for $K=s+K_1$, we get
\begin{align*}
\left| T_t \mu* \nu - \mu* \nu \right|(s+K_1) &\leq  \left( \left| \nu \right|(G)\right) \left( \| T_t \mu- \mu \|_{s+K_1'} \right) \\
&\leq  \left( \left| \nu \right|(G)\right) \left( \| T_t \mu- \mu \|_{K_1'} \right) < \frac{\epsilon}{2} \,.
\end{align*}

Now, taking the supremum over all $s$, we get that for all $t$ in the relatively dense set $P$ we have
\begin{displaymath}
\| T_t \mu* \nu - \mu* \nu \|_{K_1}  \leq  \frac{\epsilon}{2} < \epsilon  \,.
\end{displaymath}
\end{proof}

\section{On the norm almost periodicity of $\widehat{\gamma}$ for measures $\gamma$ with Meyer set support}

We can finally prove the following theorem.

\begin{theorem}\label{diff is NAP} Let $\gamma$ be a Fourier transformable measure supported inside a Meyer set. Then $\widehat{\gamma}$ is norm almost periodic.
\end{theorem}
\begin{proof}

Let $\Lambda$ be any Meyer set containing the support of $\gamma$. By \cite[Cor.~5.9.20]{NS11} there exists a CPS $(G, H, \cL)$ with $H$ compactly generated, and some compact set $W \subseteq H$ such that
\begin{displaymath}
\Lambda \subseteq \oplam(W) \,.
\end{displaymath}

Since $H$ is compactly generated, by the structure Theorem we have $H \simeq \RR^d \times \ZZ^n \times \KK$ for some $d,n$ and compact group $\KK$. We can replace then $H$ in the CPS by $\RR^d \times \ZZ^n \times \KK$ and chose the Haar measure on $\ZZ^d$ to be the counting measure, and the Haar measure on $\KK$ to be a probability measure.

Next, since $W$ is compact in $\RR^d \times \ZZ^n \times \KK$, there exists some compact $W_0 \subseteq \RR^d$ and some finite set $F \subset \ZZ^d$ such that
\begin{displaymath}
W \subseteq W_0 \times F \times K \,.
\end{displaymath}

It follows that
\begin{displaymath}
\Lambda \subseteq \oplam( W_0 \times F \times \KK ) \,.
\end{displaymath}

Now, pick some $g \in \Cc^\infty(\RR^d) \cap \cK_2(\RR^d)$ such that $g(x) =1 \forall x \in W_0$.

Let $h : \ZZ^d \times \KK$ be $h=1_F \otimes 1_{\KK}$. It is obvious that $h \in \Cc( \ZZ^d \times \KK)$ and that
\begin{displaymath}
h= h * ( 1_{0} \otimes 1_{\KK})
\end{displaymath}
which yields $h \in  \cK_2( \ZZ^d \times \KK)$.

Then, by \cite{CRS} the measure $\omega=\omega_{g \otimes h}$ is Fourier transformable and
\begin{displaymath}
\widehat{\omega}=\reallywidehat{\omega_{g \otimes h}}=\dens(\cL) \, \omega_{\widecheck{g \otimes h}} \,.
\end{displaymath}

Moreover, by construction
\begin{displaymath}
\omega(\{x \})=1 \, \forall x \in \Lambda \,.
\end{displaymath}

Then, by Lemma~\ref{L1} there exists a finite measure $\nu$ such that
\begin{displaymath}
\widehat{\gamma}= \widehat{\omega} * \nu \,.
\end{displaymath}

Let $\phi=\widecheck{g}$ and $\psi=\widecheck{h}$. Then, as $g \in \Cc^\infty(\RR^d)$ we have $\phi \in \cS(\RR^d)$. Moreover, recalling that $\widehat{\ZZ^d}$ is a compact group and that $\widehat{\KK}$ is discrete, we have
\begin{displaymath}
\psi=\widecheck{h}=\widecheck{1_F \otimes 1_{\KK}}=\widecheck{1_F} \otimes 1_{0}  \in \Cc(\widehat{\ZZ^d} \times \widehat{\KK}) \,.
\end{displaymath}

Therefore, by Lemma~\ref{lemma AP2}, the measure $ \widehat{\omega} = \omega_{ \phi \otimes \psi}$ is norm almost periodic.

As $\widehat{\omega}$ is norm almost periodic, and $\nu$ is finite, it follows now from Proposition ~\ref{conv preserves NAP} that $\widehat{\gamma}$ is norm almost periodic.

\end{proof}

As an immediate consequence we get:

\begin{corollary}\label{cor1} Let $\gamma$ be a Fourier transformable measure supported inside a Meyer set. Then each of $\left( \widehat{\gamma} \right)_{pp}, \left( \widehat{\gamma} \right)_{ac}$ and $\left( \widehat{\gamma} \right)_{sc}$ is norm almost periodic.
\end{corollary}
\begin{proof} We know by Theorem~\ref{T2} that the generalized Eberlein decomposition
\begin{displaymath}
\gamma=\gamma_{s}+\gamma_{0a}+\gamma_{0s}
\end{displaymath}
exists within the class of Fourier transformable measures with Meyer set support. The claim now follows by applying Theorem~\ref{diff is NAP} to each component.
\end{proof}

\begin{remark} It is also easy to prove that, if a measure $\mu$ is norm almost periodic, then so is each of $\mu_{pp}, \mu_{ac}, \mu_{sc}$. This gives an alternate way to deduce Corollary~\ref{cor1} from Theorem~\ref{diff is NAP}.
\end{remark}

\section{Fourier transformable measures with Meyer set support and positive definiteness}

It is well known that any linear combination of positive definite measures is Fourier transformable. An open problem in Fourier analysis, which was first mentioned implicitly in \cite[Page~29]{ARMA1}, is if the converse of this also holds:

\begin{conjecture}\label{Conj1} A measure $\mu$ on $G$ is Fourier transformable if and only if it is a linear combination of positive definite measures.
\end{conjecture}

So far, the only progress towards solving Conjecture~\ref{Conj1} we are aware of is the following result.

\begin{proposition}\label{Prop1} \cite{CRS3} Let $\mu$ be a measure supported inside some lattice $L$. Then $\mu$ is Fourier transformable if and only if $\mu$ is a linear combination of positive definite measures supported inside $L$.
\end{proposition}

In this section we extend Proposition~\ref{Prop1} to measures supported inside Meyer sets.

\smallskip

We start by proving the following preliminary lemma.

\begin{lemma}\label{L3} Let $(G,H,\cL)$ be a CPS and $W \subseteq H$ be compact set. Then, there exists a compact set $K \subseteq H$, and four positive definite measures $\omega_1,\omega_2,\omega_3, \omega_4$ supported inside $\oplam(K)$ such that, for all $x \in W$ we have
\begin{displaymath}
\omega_1(\{ x\})-\omega_2(\{x \})+i \omega_3(\{ x\}) -i \omega_4(\{ x \}) =1 \,.
\end{displaymath}

Moreover, each of $\omega_1,\omega_2,\omega_3$ and $\omega_4$, respectively, is twice Fourier transformable.
\end{lemma}
\begin{proof}

Fix an arbitrary pre-compact open set $0 \in V=-V$. Since $W+\overline{V}$ is compact, there exists some $g \in \Cc(G)$ such that $g \geq 1_{W+\overline{V}}$. Pick some $h \in \Cc(G)$ non negative such that
$\supp(h) \subseteq V$ and $\int_H h(t) \dd t =1$.

Since $V =-V$, a simple computation shows that
\begin{equation}\label{EQ5}
g*\tilde{h}(t)=1 \qquad \forall t \in W \,.
\end{equation}
Next by the polarisation identity \cite[Page~244]{MoSt}, we have
\begin{align*}
   g * \widetilde{h} &=  \frac{1}{4}
 \Bigl(\bigl[ (g+h)*\widetilde{(g+h)}
-(g-h)*\widetilde{(g-h)}\bigr] \\
&+i \bigl[ (g+i h)*\widetilde{(g+i h)} -
   (g-i h)*\widetilde{(g-i h)}\bigr] \Bigr) \,.
 \end{align*}
Since $g,h \in \Cc(G)$, we can find some compact set $K$ which contains the support of each of $ (g+h)*\widetilde{(g+h)}; (g-h)*\widetilde{(g-h)} ;(g+i h)*\widetilde{(g+i h)}$ and $(g-i h)*\widetilde{(g-i h)}$, respectively.

Define
\begin{align*}
  \omega_1 &:=\frac{1}{4} \omega_{ (g+h)*\widetilde{(g+h)}}  \\
  \omega_2 &:=\frac{1}{4} \omega_{ (g-h)*\widetilde{(g-h)}}  \\
  \omega_3 &:=\frac{1}{4} \omega_{ (g+ih)*\widetilde{(g+ih)}}  \\
  \omega_4 &:=\frac{1}{4} \omega_{ (g-ih)*\widetilde{(g-ih)}}  \,.
\end{align*}
Then, for each $1 \leq j \leq 4$ the measure $\omega_j$ is positive definite \cite[Prop.~5.8.10]{NS11}, and $\supp(\omega_j) \subseteq \oplam(K)$.

Next, by construction we have
\begin{displaymath}
\omega_1-\omega_2+i \omega_3-i \omega_4= \omega_{g*\tilde{h}} \,.
\end{displaymath}

Therefore, by Eq.~(\ref{EQ5}) we have
\begin{displaymath}
\omega_1(\{ x\})-\omega_2(\{x \})+i \omega_3(\{ x\}) -i \omega_4(\{ x \}) =1  \qquad \forall x \in W \,.
\end{displaymath}

The last claim follows from \cite[Thm.~4.12]{CRS}.
\end{proof}

\smallskip
As a consequence we get:

\begin{lemma}\label{L4} Let $(G,H,\cL)$ be a CPS and $W \subseteq H$ be compact set and let $K$ be as in Lemma~\ref{L3}. Then, for each Fourier transformable measure $\gamma$ with $\supp(\gamma) \subseteq \oplam(W)$ there exists positive definite measures $\gamma_1, \gamma_2, \gamma_3, \gamma_4$ supported inside $\oplam(K)$ such that
\begin{displaymath}
\gamma=\gamma_1-\gamma_2+i\gamma_3 -i \gamma_4 \,.
\end{displaymath}
\end{lemma}
\begin{proof}

Let $\omega_1, \omega_2, \omega_3, \omega_4$ be as in Lemma~\ref{L3}, and let
\begin{displaymath}
\omega:= \omega_1-\omega_2+i \omega_3-i \omega_4 \,.
\end{displaymath}

Then, by Lemma~\ref{L3}, $\omega$ is twice Fourier transformable, supported inside $\oplam(K)$ and satisfies
\begin{displaymath}
\omega( \{ x \}) =1 \qquad \forall x \in \oplam(W) \,.
\end{displaymath}

Therefore, by the Ping-Pong Lemma for Meyer sets (Lemma~\ref{L1} (i)), there exists a finite measure $\nu$ such that
\begin{displaymath}
\widehat{\gamma}=\widehat{\omega} * \nu \,.
\end{displaymath}

Next, write
\begin{displaymath}
\nu=\nu_1-\nu_2+i\nu_3-i \nu_4 \,,
\end{displaymath}
as a linear combination of finite positive measures. For each $1 \leq j \leq 4$ define
\begin{displaymath}
f_j := \widecheck{\nu_j} \,.
\end{displaymath}
Then, each $f_j \in \Cu(G)$ is positive definite \cite[Thm.~4.8.4]{MoSt}. Therefore, by \cite[Cor.~4.3]{ARMA1}, for all $1 \leq j,k \leq k$ the measure $f_j \omega_k$ is positive definite. 
We also have 
\begin{equation}\label{EQ6}
\supp( f_j \omega_k ) \subseteq \supp(\omega_k) \subseteq \oplam(K) \,.
\end{equation}

Next, let $1 \leq j, k \leq 4$ be arbitrary. Since $\omega_k$ is Fourier transformable, $\widehat{\omega_k}$ is translation bounded \cite[Thm.~4.9.23]{MoSt}. As $\nu_j$ is finite, the measures $\widehat{\omega_k}$ and $\nu_j$ are convolvable \cite[Lemma~4.9.19]{MoSt}. Therefore, by the twice Fourier transformability of $\omega_k$ and \cite[Lemma~4.9.26]{MoSt}, the measure $\widehat{\omega_k}*\nu_j$ is Fourier transformable and
\begin{displaymath}
\reallywidehat{\widehat{\omega_k}*\nu_j}= \reallywidehat{\widehat{\omega_k}} \widehat{\nu_j}= (f_j \omega_k)^\dagger \,.
\end{displaymath}
Since this is positive definite, it is Fourier transformable \cite[Thm.~4.11.5]{MoSt}. Therefore
\begin{displaymath}
\widehat{f_j \omega_k}= \left(\reallywidehat{(f_j \omega_k)^\dagger}\right)^\dagger = \left(\reallywidehat{\reallywidehat{\widehat{\omega_k}*\nu_j}}\right)^\dagger= \left((\widehat{\omega_k}*\nu_j)^\dagger\right)^\dagger=\widehat{\omega_k}*\nu_j \,.
\end{displaymath}

This shows that, for each $1 \leq j,k \leq 4$ the measure $f_j \omega_k$ is positive definite and
\begin{displaymath}
\widehat{f_j \omega_k}= \widehat{\omega_k}*\nu_j \,.
\end{displaymath}

Now the claim follows. Indeed, we have
\begin{align*}
\widehat{\gamma}&=\widehat{\omega} * \nu = \left( \widehat{\omega_1}-\widehat{\omega_2}+i \widehat{\omega_3}-i \widehat{\omega_4} \right)* \left( \nu_1-\nu_2+i\nu_3-i \nu_4 \right) \\
&= \left( \widehat{\omega_1}*\nu_1 +\widehat{\omega_2}*\nu_2+\widehat{\omega_3}*\nu_4+\widehat{\omega_4}*\nu_3\right)- \left( \widehat{\omega_1}*\nu_2+\widehat{\omega_2}*\nu_1+\widehat{\omega_3}*\nu_3+\widehat{\omega_4}*\nu_4 \right)\\
&+i\left( \widehat{\omega_1}*\nu_3 +\widehat{\omega_2}*\nu_4+\widehat{\omega_3}*\nu_1+\widehat{\omega_4}*\nu_2\right)- i\left( \widehat{\omega_1}*\nu_4+\widehat{\omega_2}*\nu_3+\widehat{\omega_3}*\nu_2+\widehat{\omega_4}*\nu_1\right)\\
&=\reallywidehat{\left(f_1\omega_1+f_2\omega_2+f_3\omega_4+f_4\omega_3 \right)} - \reallywidehat{\left(f_2\omega_1+f_1\omega_2+f_3\omega_3+f_4\omega_4 \right)} \\
&+i\reallywidehat{\left(f_3\omega_1+f_4\omega_2+f_1\omega_3+f_2\omega_4 \right)}-i\reallywidehat{\left(f_4\omega_1+f_3\omega_2+f_2\omega_3+f_1\omega_4 \right)} \,.
\end{align*}

Define
\begin{align*}
  \gamma_1 &:= f_1\omega_1+f_2\omega_2+f_3\omega_4+f_4\omega_3 \\
  \gamma_2 &:=f_2\omega_1+f_1\omega_2+f_3\omega_3+f_4\omega_4  \\
  \gamma_3 &:= f_3\omega_1+f_4\omega_2+f_1\omega_3+f_2\omega_4 \\
  \gamma_4 &:=f_4\omega_1+f_3\omega_2+f_2\omega_3+f_1\omega_4
\end{align*}
Then, for each $1 \leq j \leq 4$ the measure $\gamma_j$ is a sum of positive definite measures and hence positive definite. Moreover, by Eq. (\ref{EQ6}), $\supp(\gamma_j) \subseteq \oplam(K)$.

Finally, as a linear combination of positive definite measures, the measure $\gamma_1-\gamma_2+i\gamma_3 -i \gamma_4$ is Fourier transformable and
\begin{displaymath}
\reallywidehat{\gamma_1-\gamma_2+i\gamma_3 -i \gamma_4}=\widehat{\gamma} \,.
\end{displaymath}

Since the Fourier transform for measures is one to one \cite[Thm.~4.9.13]{MoSt}, we get
\begin{displaymath}
\gamma=\gamma_1-\gamma_2+i\gamma_3 -i \gamma_4 \,.
\end{displaymath}

This completes the proof.
\end{proof}

\smallskip

As a consequence we get the main result in this section.

\begin{theorem}\label{T4} Let $\mu$ be a measure supported inside a Meyer set $\Lambda$. Then $\mu$ is Fourier transformable if and only if it is a linear combination of positive definite measures.

Moreover, in this case, the positive definite measures can be chosen to be supported inside a Meyer set $\Gamma$.
\end{theorem}
\qed

\section{Applications}

Let us first summarise all the results we got in the paper.

\begin{theorem}\label{T3}  Let $G$ be a second countable LCAG, and let $\gamma$ be a Fourier transformable measure supported inside some Meyer set $\Lambda$. Then,
\begin{itemize}
\item[(i)] $\gamma$ admits an unique decomposition into three Fourier transformable measures $\gamma_{s}, \gamma_{0s}, \gamma_{0a}$ such that
\begin{align*}
\gamma=\gamma_{s} &+ \gamma_{0s} + \gamma_{0a} \\
\reallywidehat{\gamma_s}&= \left(\widehat{\gamma}\right)_{pp}  \\
\reallywidehat{\gamma_{0s}} &=  \left(\widehat{\gamma}\right)_{sc} \\
\reallywidehat{\gamma_{0a}} &=\left(\widehat{\gamma}\right)_{ac} \,.  \\
\end{align*}
\item[(ii)]    $\gamma_{s}, \gamma_{0s}, \gamma_{0a}$ are supported inside a common Meyer set.
\item[(iii)] $\left(\widehat{\gamma}\right)_{pp}, \left(\widehat{\gamma}\right)_{sc},
\left(\widehat{\gamma}\right)_{ac}$ are norm almost periodic.
\item[(iv)]  If $(G,H, \cL)$ is any CPS and $W \subseteq H$ any compact set such that $\Lambda \subseteq \oplam(W)$, then
\begin{displaymath}
\supp(\gamma_{s}), \supp(\gamma_{0a}), \supp(\gamma_{0s}) \subseteq \oplam(W) \,.
\end{displaymath}
\item[(v)]  If $(G,H, \cL)$ is any CPS and $W \subseteq H$ any compact set such that $\Lambda \subseteq \oplam(W)$, and $h \in \cK_2(H)$ is such that $h(x)=1$ for all $x \in W$ then, there exists a finite measure $\nu$ such that
\begin{align*}
\left(\widehat{\gamma}\right)&= \left(\sum_{(y,y^\star) \in \cL^0} \widecheck{h}(y^\star) \delta_y \right)* \nu =\omega_{\widecheck{h}}*\nu \\
\left(\widehat{\gamma}\right)_{pp} &=\left(\sum_{(y,y^\star) \in \cL^0} \widecheck{h}(y^\star) \delta_y \right)* \nu_{pp}=\omega_{\widecheck{h}}*\nu_{pp} \\
\left(\widehat{\gamma}\right)_{sc} &=\left(\sum_{(y,y^\star) \in \cL^0} \widecheck{h}(y^\star) \delta_y \right)* \nu_{sc}=\omega_{\widecheck{h}}* \nu_{sc}  \\
\left(\widehat{\gamma}\right)_{ac} &=\left(\sum_{(y,y^\star) \in \cL^0} \widecheck{h}(y^\star) \delta_y \right)* \nu_{ac}=\omega_{\widecheck{h}}*\nu_{ac}  \,. \\
\end{align*}
\item[(vi)] There exists a Meyer set $\Gamma$ and four positive definite measures $\gamma_1, \gamma_2, \gamma_3, \gamma_4$ supported inside $\Gamma$ such that
\begin{displaymath}
\gamma=\gamma_1-\gamma_2+i\gamma_3 -i \gamma_4 \,.
\end{displaymath}
\end{itemize}
\end{theorem}
\qed

Next, we list some applications to the diffraction.

\begin{corollary}\label{cor 2} Let $\omega$  be weighted Dirac comb supported inside a Meyer set, and let $\gamma$ be an autocorrelation of $\omega$. Then each of the diffraction spectral measures $\left( \widehat{\gamma} \right)_{pp}, \left( \widehat{\gamma} \right)_{ac}$ and $\left( \widehat{\gamma} \right)_{sc}$ is norm almost periodic.
\end{corollary} \qed

\begin{corollary}\label{cor 3} Let $\omega$  be weighted Dirac comb supported inside a Meyer set, and let $\gamma$ be an autocorrelation of $\omega$. Then each of diffraction spectral measure $\left( \widehat{\gamma} \right)_{pp}, \left( \widehat{\gamma} \right)_{ac}$ and $\left( \widehat{\gamma} \right)_{sc}$ is either trivial or has relatively dense support.
\end{corollary}\qed

\begin{remark} In \cite{NS1} we have proven that given a Meyer set $\Lambda$, each of $\left( \widehat{\gamma} \right)_{pp}, \left( \widehat{\gamma} \right)_{c}$ of the diffraction spectra of $\Lambda$ is strong almost periodic and hence, it is either trivial or has relatively dense support. These results are an immediate consequence of Corollary~\ref{cor 2} and Corollary~\ref{cor 3}.
\end{remark}

\begin{remark} In \cite{NS2} we have proven that given a be weighted Dirac comb supported inside a Meyer set, its pure point spectrum $\left( \widehat{\gamma} \right)_{pp}$ is a sup-almost periodic measure (see \cite[Sect.~5.3]{NS11} for the definition and properties of sup-almost periodicity) and that, for each $0< a < \widehat{\gamma}(\{ 0 \})$ the set
\begin{displaymath}
\mathcal{I}(a):= \{ \chi \in \widehat{G} | \widehat{\gamma}(\{ \chi \}) \geq 0 \}
\end{displaymath}
is a Meyer set.

Both of these are immediate consequences of Corollary~\ref{cor 2} and some results of \cite{NS11}. Indeed, since $\left( \widehat{\gamma} \right)_{pp}$ is norm-almost periodic, it is sup-almost periodic by \cite[Lemma~5.3.6]{NS11}. Next, since $\left( \widehat{\gamma} \right)_{pp}$ is a sup-almost periodic pure point measure, by \cite[Thm.~5.2]{NS11} there exists a CPS $(\widehat{G}, \widehat{H}, \cL)$ and some $h \in \Cz(\widehat{H})$ such that
\begin{displaymath}
\left( \widehat{\gamma} \right)_{pp} =\omega_h \,.
\end{displaymath}
Since $\gamma$ is positive definite, $\widehat{\gamma}$ and hence $\left( \widehat{\gamma} \right)_{pp}$ is positive. Then $h$ is also positive, by the denseness of the second projection.

Therefore
\begin{displaymath}
\mathcal{I}(a)= \oplam(W_a) \,,
\end{displaymath}
where $W_a:=\{ y \in \widehat{H} | h(y) \geq a \}$ is compact and has non-empty interior, as $h \in \Cz(\widehat{H})$ and $h \nequiv 0$.

It follows that for all $0 <a < \widehat{\gamma}(\{ 0 \})$ the sets $\mathcal{I}(a)$ are model sets in the same CPS.
\end{remark}

\begin{corollary} Let $\omega$  be weighted Dirac comb supported inside a Meyer set, and let $\gamma$ be an autocorrelation of $\omega$. Then, for each pre-compact Borel set $B \subseteq \widehat{G}$ there exists a relatively dense set $P \subseteq \widehat{G}$ such that, for all $t \in P$ we have
\begin{align*}
  \left| \widehat{\gamma}(t+B)- \widehat{\gamma}(B) \right|  &<\epsilon  \\
  \left| \left( \widehat{\gamma} \right)_{pp}(t+B)- \left( \widehat{\gamma} \right)_{pp}(B) \right|  &<\epsilon  \\
  \left| \left( \widehat{\gamma} \right)_{ac}(t+B)- \left( \widehat{\gamma} \right)_{ac}(B) \right|  &<\epsilon  \\
  \left| \left( \widehat{\gamma} \right)_{sc}(t+B)- \left( \widehat{\gamma} \right)_{sc}(B) \right|  &<\epsilon  \\
\end{align*}
\end{corollary}
\begin{proof}

Let $K$ be any compact set with non-empty interior such that $B \subseteq K$.

Let us first observe that if $\mu$ is a norm almost periodic measure, and $\epsilon >0$, then the set of common almost periods of $\mu, \mu_{pp}, \mu_{ac}$ and $\mu_{sc}$ is relatively dense.

Indeed, by \cite[Thm.~14.22]{HeRo}, we have
\begin{displaymath}
|\mu|=|\mu_{\mathsf{pp}}|+| \mu|_{\mathsf{ac}}+|\mu|_{\mathsf{sc}} \,.
\end{displaymath}

It follows immediately from here that, for each $\alpha \in \{\mathsf{pp}, \mathsf{ac}, \mathsf{sc}\}$ we have
\begin{displaymath}
\| \mu_{\alpha} \|_{K} \leq\| \mu \|_K  \leq \| \mu_{\mathsf{pp}} \|_K+\| \mu_{\mathsf{ac}} \|_K+\| \mu_{\mathsf{sc}} \|_K \,. \\
\end{displaymath}

This implies that, for each $\alpha \in \{\mathsf{pp}, \mathsf{ac}, \mathsf{sc}\}$ we have
\begin{displaymath}
P_{\epsilon}^K ( \mu) \subseteq P_{\epsilon}^K ( \mu_{\alpha}) \,.
\end{displaymath}

Thus, as $\mu$ is norm almost periodic, the relatively dense set $P_{\epsilon}^K ( \mu)$ is a common set of $\epsilon$ norm-almost periods for $\mu, \mu_{pp}, \mu_{ac}$ and $\mu_{sc}$.

Now, for each $-t$ in this relatively dense set we have
\begin{displaymath}
 \left| \widehat{\gamma}(t+B)- \widehat{\gamma}(B) \right| \leq  \left| T_{-t}\widehat{\gamma}- \widehat{\gamma} \right|(B) \leq  \left| T_{-t}\widehat{\gamma}- \widehat{\gamma} \right|(K) \leq  \| T_{-t}\widehat{\gamma}- \widehat{\gamma} \|_{K} < \epsilon \,,
\end{displaymath}
and similarly,
\begin{displaymath}
 \left| \left(\widehat{\gamma}\right)_{\alpha}(t+B)- \left(\widehat{\gamma}\right)_{\alpha}(B) \right| \leq   \| T_{-t}\left(\widehat{\gamma}\right)_{\alpha}- \left(\widehat{\gamma}\right)_{\alpha}\|_{K} < \epsilon \qquad \forall \alpha \in \{\mathsf{pp}, \mathsf{ac}, \mathsf{sc}\} \,.
\end{displaymath}
\end{proof}

\subsection*{Acknowledgments} We are grateful to Christoph Richard, Franz G\"{a}hler and Michael Baake for many insightful discussions which inspired this manuscript. The manuscript was completed when the author visited Bielefeld University, and the author is grateful to the mathematics department for hospitality. This research stay was partially supported by the Simons Foundation and by the Mathematisches Forschungsinstitut Oberwolfach. The work was partially supported by NSERC with grant 03762-2014. We are greatly thankful for all the support.

\end{document}